\documentclass{fundam}

\usepackage{graphicx}
\usepackage{makeidx}
\usepackage{latexsym} 
\usepackage{url}

\newcommand{\msum}{\mbox{
  \begin{picture}(5,11)
  \put(-4.6,0){ $\cup$}
  \put(-1,1.5){\tiny $\cup$} 
  \end{picture}
  }}

\newcommand{\psum}{\mbox{
  \begin{picture}(5,11)
  \put(-4.6,0){ $\cup$}
  \put(-1,1.9){\tiny $+$} 
  \end{picture}
  }}

\newcommand{\pa}{\mbox{
  \begin{picture}(3,11)(3,0)
  \put(2,0){\small $\parallel$}
  \end{picture}
  }}

\newcommand{\konk}{
   \begin{picture}(8,11)(0,0)
    \put(-2,0){ $\circ$}
   \end{picture}
   }

\newtheorem{Remark}{\sc Remark}[section] 

\newtheorem{conjecture}[Remark]{\bf Conjecture}

\begin{document}

\setcounter{page}{63}
\publyear{22}
\papernumber{2142}
\volume{188}
\issue{2}

   \finalVersionForARXIV

\title{On Completeness of Cost Metrics and Meta-Search Algorithms\\ in \$-Calculus}

\author{Eugene Eberbach\thanks{retired Professor of Practice, RPI Hartford, CT. This paper has been written in memories
            of Professor Peter Wegner from Brown University and Professor Mark Burgin from University of California Los Angeles
            with whom I had the honor and pleasure to cooperate for several years.}\thanks{Address for correspondence: Department of Engineering and Sci.,  Rensselaer Polytechnic Institute, Hartford, CT, USA.  \newline \newline
                    \vspace*{-6mm}{\scriptsize{Received July 2021; \ accepted February  2023.}}}
\\
Department of Engineering and Science\\
Rensselaer Polytechnic Institute, Hartford, CT, USA\\
eeberbach@gmail.com
}
\maketitle

\runninghead{E. Eberbach}{On Completeness of Costs and Meta-Search}

\vspace*{-6mm}
\begin{abstract}
In the paper we define three new complexity classes for Turing Machine undecidable problems inspired by the famous Cook/Levin's NP-complete complexity class for intractable problems. These are U-complete (Universal complete), D-complete (Diagonalization complete) and H-complete (Hypercomputation complete) classes. In the paper, in the spirit of Cook/Levin/ Karp, we started the population process of these new classes assigning several undecidable problems to them. We justify that some super-Turing models of computation, i.e., models going beyond Turing machines, are tremendously expressive and they allow to accept arbitrary languages over a given alphabet including those undecidable ones. We prove also that one of such super-Turing models of computation - the \$-Calculus, designed as a tool for automatic problem solving and automatic programming, has also such tremendous expressiveness. We investigate also completeness of cost metrics and meta-search algorithms in \$-calculus.

\end{abstract}

\section{Introduction}

This paper is on problem solving.  If problem solving was everything what humans were doing in their life, then the paper would be on everything. However, it is obvious that problem solving although dominating our activities and very important does not cover all human activities. Thus the paper for sure is not on the theory of everything (an unachievable goal), but at most on its approximation. Note also that economy, evolutionary computing, mathematics and Turing machines (TMs), although very general and powerful also fail as theory of everything. This has been stressed and reminded again in Conclusion section.

Problems can be either solvable or unsolvable (called also undecidable) using a specific model/ theory. In computer science, Turing machines form such dominating and most popular model for problem solving.
Some problems are Turing machine solvable and some not. In particular, our paper intends to provide a new approach to deal with Turing machine undecidable problems. The typical belief is that  proving that a specific problem is TM-undecidable stops any attempt to solve that problem and that is the end of the story. On the other hand, we are convinced that this is only the beginning. First of all, we may decide special instances (or perhaps even almost all instances) of the undeciable problem. For example, if we have the probability distribution of input instances, perhaps randomized techniques may help to estimate which inputs are decidable. Secondly, we can approximate the solutions and we may decide the specific instances either in a finite number of steps or asymptotically in the infinity. There are other approaches possible to deal with undecidability too (see, e.g., the infinity, evolution and interaction principles from next section), and all above looks like an excellent and exiting new venue for many years of fruitful research to come.

\medskip
This paper is organized as follows. In section 2, we overview some basic notions  related to Turing machine
problem solving and computations going beyond Turing machines.
We call that a super-Turing computation or hypercomputation. In particular, we introduce 3 complexity classes in a new way compared to first attempt to do so in \cite{eber15}, and inspired by the famous Cook/Levin NP-complete complexity class. We present some examples of TM-undecidable problems and several super-Turing models of computation. We justify that some of such models allow to accept arbitrary languages (including those undecidable ones) over a given alphabet. In section 3, we outline the
\$-calculus super-Turing model of computation based on cost directed $k\Omega$-search as a tool for automatic problem solving and automatic programming. In section 4, we prove the expressiveness of \$-calculus. In section 5, we investigate the problem of completeness of cost functions, and in section 6 - the completeness of meta-search algorithms.
Section 7 contains conclusions and problems to be solved in the future.

\section{Super-Turing computation in capsule}

Turing Machines \cite{turing36,turing39} and algorithms are two fundamental concepts of computer science
and problem solving. Turing Machines describe the limits of problem solving using conventional recursive algorithms,
and laid the foundation of current computer science in the 1960s.
TMs made imprecise definitions of algorithms mathematically more precise and formal and they led to Kleene's Turing Thesis that every algorithm can be represented in the form of Turing machine.
\eject

Note that there are several other models of algorithms, called super-recursive algorithms, that can compute more than Turing Machines, using hypercomputational/superTuring models of computation \cite{burgin05,syropoulos08}.

It turns out that (TM) {\em undecidable problems} cannot
be solved by TMs and {\em intractable problems} are solvable,
but require too many resources (e.g., steps or memory). For undecidable problems effective
recipes do not exist -
they are covered by several classes of nonrecursive algorithms
(two outer rings from Figure 1).
 On the other hand, for intractable problems algorithms exist, but
running them on a deterministic Turing Machine,
requires an exponential
amount of time (the number of elementary moves of the TM)
as a function of the TM input.

\begin{figure}[!h]
\vspace*{-2mm}
\begin{center}
\includegraphics[width=11.6cm]{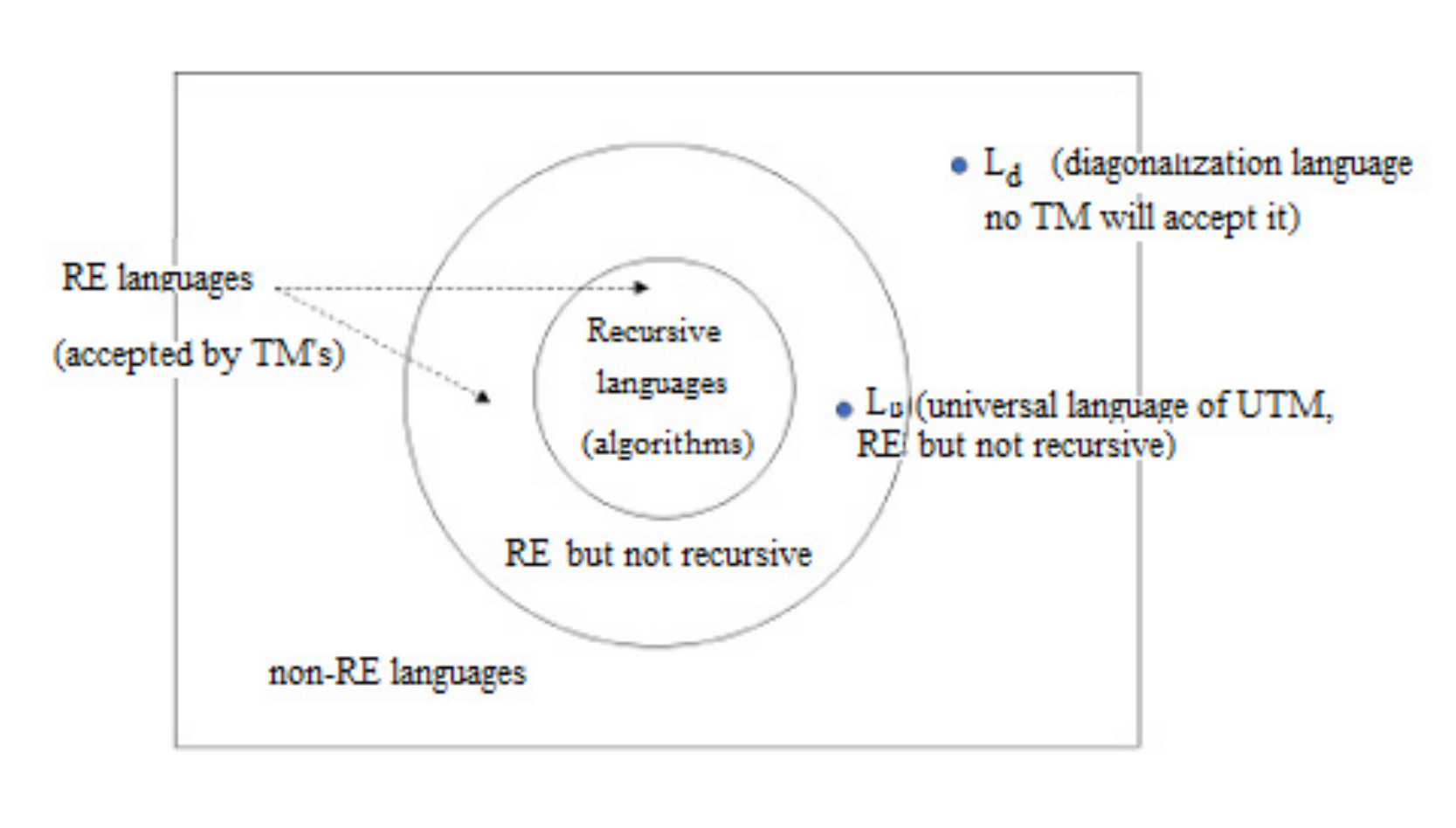}
\end{center}\vspace*{-11mm}
\caption{Relation between Recursive, Re and non-Re Languages}\vspace*{2mm}
\end{figure}

We use the simplicity of the TM model to prove formally that there are specific
problems (languages) that the TM cannot solve \cite{hopcroft01}.
Solving the problem is equivalent to decide whether a string belongs to the language.
A problem that cannot be solved by computer (Turing machine) is called
{\em undecidable} (TM-undecidable).
The class of languages accepted by Turing machines are called
{\em recursively enumerable (RE-) languages}.
For RE-languages, TM can accept the strings in the language but cannot tell
for certain that a string is not in the language.

\medskip
There are two classes of Turing machine unsolvable languages (problems):
\begin{description}
\item{\em recursively enumerable but not recursive (RE-nonREC)}  - TM  can accept the strings in the language
but cannot tell for certain that a string is not in the language (e.g., the language of the universal
Turing machine, or Post's Correspondence Problem languages).
A language is decidable but its complement is undecidable.
\item{\em non-recursively enumerable (non-RE)} - no TM can even recognize the members of the language
in the RE sense (e.g., the diagonalization language). Neither a language nor
its complement is decidable, or a language is undecidable
but its complement is decidable.
\end{description}

Decidable problems have a (recursive) algorithm, i.e., TM halts whether or not it accepts its input.
Decidable problems are described by {\em recursive languages}.
Recursive algorithms as we know are associated with the class of recursive languages,
a subset of recursively enumerable languages for which we can construct its accepting TM.
For recursive languages, both a language and its complement are decidable.

Turing Machines are used as a formal model of classical (recursive)  algorithms.
An algorithm should consist of a finite number of steps, each having well defined
and implementable meaning.
We are convinced that computer computations are not restricted to such restrictive
definition of algorithms only. If we allow for an infinite number of steps (e.g., reactive programs)
and/or not well defined/implementable meaning
of each step (e.g., an Oracle), we are in the class of super-recursive algorithms \cite{burgin05}.

In \cite{denning11}, Peter Denning expressed concerns that developments in non-terminating computation, analog computation, continuous computation, and natural computation may require rethinking the basic definitions of computation. It is further stated that computation is the process that the machine or algorithm generates. In analogies: the machine is a car, the desired outcome is the driver's destination, and the computation is the journey taken by the car and driver to the destination.

The Turing machine model is not the most convenient for many domains even though theoretically it is equivalent with every convenient model in its ability to define functions. In solving real problems, we work with computational models (or invent new ones) that are convenient and appropriate for the domain \cite{denning11,aho11}.

Such new types of computation and computational models are called very often hypercomputation (or super-Turing computation) and hypercomputational (or super-Turing) models of computation.

\begin{definition}[On super-Turing computation]
By {\em super-Turing computation} (also called {\em hypercomputation}) we mean any computation that cannot be carried out
by a Turing machine as well as any (algorithmic) computation carried out by a Turing
machine.
\end{definition}

The above definition is consistent with \cite{syropoulos08,burgin05}) (see also Wikipedia's definition for comparison) that states  that hypercomputation or super-Turing computation refers to models of computation that go beyond, or are incomparable to, Turing computability. This includes various hypothetical methods for the computation of non-Turing-computable functions, following super-recursive algorithms (see \cite{syropoulos08,burgin05}).

\medskip
Super-Turing  models derive their higher than the TM expressiveness using three principles {\em interaction, evolution}, or {\em infinity}:
\begin{itemize}
\itemsep=0.9pt
\item
In the {\em interaction principle} the model becomes open and the agent
interacts with either a more expressive component or with an infinite many components.
\item
In the {\em evolution principle}, the model can evolve to a more expressive one
using non-recursive variation operators.
\item
In the {\em infinity principle}, models can use unbounded resources:
time, memory, the number of computational elements, an unbounded initial configuration,
an infinite alphabet, etc.
\end{itemize}
The details can be found in \cite{eber03c,eber04a}.
\eject

\begin{definition}[On time/space complexity]
We can define four classes of problems/languages to decide strings in the language in the order of increasing hardness and complexity:

\begin{enumerate}
{\em
\item
{\bf p-decidable problems:} The number of steps/memory cells is polynomial in the problem size and problems will be called {\em polynomial decidable}
({\em p-decidable}).
\item
{\bf e-decidable problems:} The number of steps/memory cells is exponential  in the problem size and problems will be called {\em exponentially decidable}
({\em e-decidable}).
\item
{\bf a-decidable problems:} The number of steps/memory cells is infinite but computed in finite time/finite number of cells, i.e., {\em asymptotically/limit
decidable} ({\em a-decidable})
(analogy: convergent infinite series,
mathematical induction, computing infinite sum in definite integral).
\item
{\bf i-decidable problems:} The number of steps/memory cells is infinite and requires infinite time/ infinite number of cells
to decide strings in the problem size, i.e., {\em infinitely decidable} ({\em i-decidable}) (undecidable in the finite sense).
}
\end{enumerate}
\end{definition}
The classical complexity theory usually covers classes (1) as easy/tractable and class (2) as intractable problems.
Class (3) is
an intermediate class because although technically it requires an infinite number of steps
(or memory cells for space complexity), we can find
a solution in the limit using finite resources. It is represented by convergent in infinity subset of
Inductive Turing Machines,
anytime algorithms,
evolutionary algorithms, or \$-calculus.
Class (4) requires infinite resources and is unsolvable by Turing Machine
(but solvable by hypercomputers using infinite resources). Classes (2) , (3) and (4) cover non-polynomial algorithms,
classes (3) and (4)
belong to super-recursive algorithms \cite{burgin05} (see also Wikipedia). Class (1) and (2) cover recursive algorithms.

Obviously, undecidable problems are characterized by computations growing faster than exponentially (i.e., hyperexponentially) or they may have even an infinite computational complexity (e.g., an enumerable or real numbers infinity type).

Note that we have in reality an infinite hierarchy of infinities/cardinalities in mathematics/set theory \cite{kuratowski77}, whereas computer science considers only typically enumerable infinity (denoted by  $\omega$, $\alpha$   or  $\aleph_0$) with some extension to true real numbers (denoted by $c$ or $\aleph_1$) represented by analog computers, neural networks operating on real numbers or evolution strategies operating on vectors of real numbers.

Traditional complexity theory deals with various type of decidable algorithms (see e.g., \cite{kleinberg06}), i.e., graph, greedy, divide and conquer, dynamic programming, network flow, NP-complete, PSpace, approximation, local search, and randomized algorithms. All of them have polynomial or exponential complexities. Algorithms than run forever \cite{kleinberg06} and super-recursive algorithms \cite{burgin05} require an extension of complexity theory to infinite cases (nobody did it so far for enumerable $\aleph_0$  or non-enumerable infinities $\aleph_1$, $\aleph_2$, $\aleph_3$,...). For example, local search, also known as hill climbing, typically allows to find local optima only. However, evolutionary algorithms if the search is complete and uses elitist selection allows in infinity to reach global optimum. Then, we do not solve a given problem in polynomial time nor in exponential time, but, perhaps, in infinite time (classes a-decidable and i-decidable of algorithms). But if infinite, what type of infinity we are talking about, $\aleph_0$ or $\aleph_1$?  Or something else?

It is necessary to make a distinction between recursively undecidable problems and super-recursively undecidable problems. At this moment, the author hypothesizes that H-complete class covers super-recursively undecidable problems. Whether all of them? Probably not. This is an open research problem. The author's belief is that assuming that Figure 1 covers all possible problems then Turing machine  recursively undecidable problems belong to two outer rings, and super-recursively undecidable problems belong only to the outer ring.

Note also that the granularity of outer ring in Figure 1 is not sufficient, because currently contains both D-complete, H-complete and complement of U-complete languages. In the future, the outer ring has to be partitioned further.

\medskip
In \cite{eber03c,eber04a,syropoulos08,pesonen11}, several super-Turing models have been discussed and overviewed. An incomplete list includes:
\begin{itemize}
\itemsep=0.95pt
\item {\em Turing's o-machines, c-machines and u-machines} (Turing A.) - they use help of Oracle (o-machines) or human operator (c-machines),
or they form an unorganized neural network that may evolve by genetic algorithms or reinforcement learning (u-machines),
\item
{\em Cellular automata} (von Neumann J.) - an infinite number of discrete finite automata cells in a regular grid,
\item
{\em Discrete and analog neural networks} (Garzon M., Siegelmann H.) - a potentially infinite number of discrete neurons or neurons with true real-valued inputs/outputs,
\item
{\em Interaction Machines} (Wegner P.) - they interact with other machines sequentially or in parallel by infinite multiple streams of inputs and outputs,
\item
{\em Persistent Turing Machines} (Goldin D.) - they preserve contents of memory tape from computation to computation,
\item
{\em Site and Internet Machines} (van Leeuwen J., Wiedermann J.) - they have input/output ports that allow to interact with an environment or Oracle and communicate by infinite streams of messages,
\item
{\em The $\pi$-calculus} (Milner R.) - potentially an infinite number of agents interacting in parallel by message-passing,
\item
{\em The \$-calculus} (Eberbach E.) - potentially an infinite number of agents interacting in parallel by message-passing and searching for solutions by built-in $k\Omega$-optimization meta-search that may evolve,
\item
{\em Inductive Turing Machines} (Burgin M.) - they may continue computation after providing the results in a finite time,
\item
{\em Infinite Time Turing Machines} (Hamkins J.D.) - they allow an infinite number of computational steps,
\item
{\em Accelerating Turing Machines} (Copeland B.J.) - each instruction requires a half of the time of its predecessor's time forming a geometric convergent series,
\item
{\em Evolutionary Turing Machines} (Eberbach E.) and {\em Evolutionary Automata} (Eberbach E., Burgin M.) - they use an infinite chain of abstract automata that may evolve in successive generations and communicate by message-passing (an output becomes an input to a next generation).
\end{itemize}

It is interesting that we have more undecidable/unsolvable problems (represented by the cardinality of real numbers - an uncountable infinity) than decidable ones (represented by the cardinality of natural numbers - an enumerable infinity). This is caused by the fact that decidable problems are modeled by Turing Machines and we have only an infinte enumerable number of Turing Machines possible (see, e.g., \cite{hopcroft01}).

The new complexity classes will be defined in the order of growing undecidability. Note that we concentrated on time computational complexities inspired by NP-complete problems class. An analogous classification can be provided based on memory computational complexities, e.g., inspired by PSPACE-complete problems class.

\medskip
Before we will define 3 new complexity classes for Turing Machine undecidable problems, we will present a few examples of typical unsolvable problems.

The Universal TM simulates the work of arbitrary TM. Its language is RE but not recursive.
\begin{definition}[On the Universal TM Language]
The {\em universal TM language} $L_u$ (of the Universal Turing Machine) UTM $U$ consists of the set
of pairs $(M,w)$, where $M$ is a binary encoding of TM and $w$ is its binary input. The UTM $U$ accepts $(M,w)$ iff
TM $M$ accepts $w$ \cite{hopcroft01}.
\end{definition}

The diagonalization language $L_d$ is an example of the language that is believed
even more difficult in solvability than $L_u$, i.e., language of UTM accepting words $w$ for arbitrary
TM $M$. $L_d$ is non-RE, i.e., it does not exist any TM accepting it.

\begin{definition}[On the Diagonalization Language]
The {\em diagonalization language} $L_d$  consists of all strings $w$ such that TM $M$ whose
code is $w$ does not accept when given $w$ as input \cite{hopcroft01}.
\end{definition}

The existence of the diagonalization language that cannot be accepted by any TM
is proven by the diagonalization table with
``dummy'', i.e., not real/true values. Of course, there are many diagonalization
language encodings possible that depend how transitions of TMs are encoded.
This means that there are infinitely many different $L_d$ language instances (but
nobody wrote a specific example of $L_d$).
Solving the halting problem of UTM, can be used for a ``constructive''
proof of
the diagonalization language i-decidability demonstrating all strings belonging to the language.

\begin{definition}[On Nonempty and Empty TM Languages]
The {\em Nonempty TM language} $L_ne$ consisting of all binary encoded TMs whose language is not empty, i.e., $L_{ne} = \{ M\; | \;L(M) \neq \emptyset \}$ is known to be recursively enumerable but not recursive, and its complement - the {\em Empty TM language} $L_e = \{ M\; | \;L(M) =  \emptyset \}$ consisting of all binary encoded TMs whose language is empty is known to be non-recursively enumerable \cite{hopcroft01}.
\end{definition}

\begin{definition}[The Post Correspondence Problem (PCP)]
The TM undecidable {\em Post Correspondence Problem (PCP)} \cite{post46} asks, given two lists of the same number of strings over the same alphabet, whether we can pick a sequence of corresponding strings from the two lists and form the same string by concatenation.
\end{definition}

\begin{definition}[On Busy Beaver Problem (BBP)]
The TM undecidable {\em Busy Beaver Problem (BBP)} \cite{rado62} considers a deterministic 1-tape
Turing machine with unary alphabet $\{1\}$ and tape alphabet $\{1,B\}$, where $B$
represents  the tape blank symbol. TM starts with an initial empty tape and accepts by
halting. For the arbitrary number of states $n=0,1,2,...$ TM tries to compute  two functions:
the maximum number of 1s written on tape before halting (known as the busy beaver
function $\Sigma(n))$, and the maximum number of steps before halting (known as the
maximum shift function $S(n))$.
\end{definition}

In \cite{bringsjord12}, the author (inspired by the recent famous or infamous - pending on the point of view) the world economy crisis from 2008) introduced, partially for ``fun'', two other undecidable problems, related to BBP, namely the Economy Collapse Problem (ECP) and the Economy Immortality Problem (EIP).

\begin{definition}[The Economy Collapse Problem (ECP)]
Let $ECP(n)$ be the maximum amount of time for which any economy with $n$ states can function without collapsing \cite{bringsjord12}.  Here collapse can be equated with the corresponding Turing machine $M$
halting when started on a blank tape.  Compute $ECP(n)$ for
arbitrary values of $n$.
\end{definition}

\begin{definition}[The Economy Immortality Problem (EIP)]
Let $EIP$ represent economies that never collapse, i.e., are
immortal \cite{bringsjord12}.  For arbitrary values of $n$ decide whether any economy with $n$ states is immortal.
\end{definition}

\subsection{Undecidability complexity classes: U-complete, D-complete and H-complete problems}

Now we are ready to introduce 3 new classes of TM undecidable problems, inspired by the NP-complete class definition.

\begin{definition}[On U-complete languages]
We say a language $L$ is {\em U-complete (Universal Turing Machine complete)} iff
\begin{enumerate}
\item
Any word $w$ can be decided in a finite number of steps if  $w \in L$, or it requires an infinite number of steps if  $w \notin  L$ (semi-decidability condition).
\item
For any language $L'$ satisfying (1) there is p-decidable, or e-decidable reduction of $L'$ to $L$ (completeness condition).
\end{enumerate}
\end{definition}

\begin{corollary}
U-complete languages belong to the {\em RE-nonREC} class.
\end{corollary}

Examples of U-complete languages include $L_u$ (a basic representative to call the whole class), PCP, $L_{ne}$, BBP, ECP, planning problem, optimization problem.
\eject

The {\em U-hard languages}, a superset of U-complete languages, satisfy only completeness condition from the above definition.

Note that U-complete semi-decidable languages resemble NP-complete class - if you know the solution (the string in the language) you can decide in a finite number of steps that the string is accepted, similar like in NP class you can decide in a polynomial time about string acceptability. For this reason it starts a hybrid  ``easiest class'' in the hierarchy of TM undecibale problems, followed by D-complete and H-complete languages.

Note also that proving that any member of U-complete class is decidable, it will break down the undecidability of all remaining members of U-complete class (similar like it is with the NP-complete class).

\begin{definition}[On D-complete languages]
We say a language L is {\em D-complete (Diagonalization complete)} iff
\begin{enumerate}
\itemsep=0.95pt
\item
Any word $w$ from $L$ cannot be decided in a finite number of steps (undecidability condition).
\item
For any language $L'$ satisfying (1) there is p-decidable, or e-decidable reduction of $L'$ to $L$ (completeness condition).
\end{enumerate}
\end{definition}

\begin{corollary}
D-complete languages belong to the {\em non-RE} class.
\end{corollary}

Examples of D-complete languages include $L_d$ (a basic representative to call the whole class), $L_e$, EIP, complement of $L_d$, complement of $L_u$, complement of BBP.

The {\em D-hard languages}, a superset of D-complete languages, satisfy only completeness condition from the above definition.

\begin{definition}[The hyper-diagonalization language]
The {\em hyper-diagonalization language} $L_{hd}$ consists of all strings $w$ such that TM $M$ whose code is $w$ will not  accept even in an infinite number of steps when given $w$ as input.
\end{definition}

\begin{definition}[On H-complete languages]
We say a language $L$ is {\em H-complete (Hypercomputation complete)} iff
\begin{enumerate}
\itemsep=0.95pt
\item
Any word $w$ from or outside of $L$ cannot be decided in an infinite number of steps (hyper-undecidability condition).
\item
For any language $L'$ satisfying (1) there is an a-decidable or i-decidable reduction of $L'$ to $L$ (completeness condition).
\end{enumerate}
\end{definition}

\begin{corollary}
H-complete languages belong to the {\em non-RE} class.
\end{corollary}

A canonical representative of this class is $L_{hd}$ .

The {\em H-hard languages}, a superset of H-complete languages, satisfy only completeness condition from the above definition.

\subsection{Expressiveness of Super-Turing Models of computation}

This subsection on expressiveness of Super-Turing Models of computation is based on \cite{eber22}.

\begin{definition}[On terminal mode]
A word w is accepted in the {\em terminal mode} of the automaton $E$ if given the word $ w$ as input to the automaton $E$, there is a number $n$ such that the automaton $A[n]$ from $E$ comes to an accepting state.
\end{definition}

\begin{definition} [On terminal languages]
The terminal language $TL(E)$ of the automaton $E$ is the set of all words accepted in the terminal mode of the automaton $E$.
\end{definition}

In \cite{burgin12,wegner12} it has been proven that Evolutionary Automata (e.g., Evolutionary Turing Machines or Evolutionary Finite Automata) and Interaction Machines accept arbitrary languages over a given alphabet.

\begin{theorem} Terminal languages of Evolutionary Automata and Interaction Machines coincide with the class of all languages in the alphabet $X$.
\end{theorem}

We believe that analogous proofs can be derived for \$-calculus (see next section for the details), $\pi$-calculus (pending that replication operator allows for infinity), cellular automata (extended to random automata networks, where each cell may represent a different finite state automaton), neural networks, Turing u-machines (pending that they allow for an infinite number of nodes), i.e., models where we can derive the sequence of components inheriting all needed information from their predecessors, i.e., we can repeat essentially the proofs for evolutionary automata and interaction machines \cite{burgin12,wegner12}.  Thus we will write, skipping the proofs, the following.

\begin{conjecture} Terminal languages for \$-calculus, $\pi$-calculus, cellular automata generalized to random automata networks, neural networks and Turing u-machines coincide with the class of all languages in the alphabet $X$.\hfill $\Box$
\end{conjecture}

We can also safely assume that models based on Oracles, i.e., Turing o-machines \cite{turing39} and Site and Internet Machines can also accept arbitrary languages over a given alphabet, because of the power and magic of the Oracle ``black box''.

\begin{theorem} Terminal languages of o-machines and Site and Internet Machines coincide with the class of all languages in the alphabet $X$.
\end{theorem}
From Theorems 2.19, 2.20 and Conjecture 2.1 we can derive immediately the conclusion.

\begin{corollary} Expressiveness of  \$-calculus, $\pi$-calculus, cellular automata, neural networks, Turing o-machines and u-machines, Evolutionary Automata and Interaction Machines is the same and allow to accept all languages over a given finite alphabet. \hfill $\Box$
\end{corollary}

It is not clear at this moment how to classify expressiveness of Infinite Time Turing Machines and Accelerating Turing Machines - simply, the conditions of an infinite number of steps or doubling the speed of each successive step alone seem not be sufficient to prove that those models can accept all languages over a given alphabet. Similarly, we do not have enough details on c-machines, because they were only briefly mentioned in the original paper on Turing machines \cite{turing36}. Also we cannot properly classify at this moment the expressiveness of Inductive Turing Machines and Persistent Turing Machines in the form of the stand-alone components. However, it is clear that they, as components of Evolutionary Automata or Interaction Machines, may achieve such enormous expressiveness of their hosts.

From the above we conclude the following.

\begin{corollary}  Turing o-machines and u-machines, Site and Internet Machines, \$-calculus, $\pi$-cal\-cu\-lus, cellular automata, neural networks, Evolutionary Automata and Interaction Machines accept all U-complete, D-complete and H-complete languages. \hfill $\Box$
\end{corollary}

\section{The \$-calculus super-Turing model of computation}

We will concentrate our investigation on one of super-Turing models of computation, the \$-calculus.

The \$-calculus (read: cost calculus) is a process algebra using anytime algorithms for automatic problem solving and automatic programming
targeting {\em intractable} and {\em undecidable problems}
\cite{eber05a,eber07,eber08}.
Note that automatic problem solving and automatic programing are in a general case undecidable, however they can and have to be approximated to solve important problems in computer science and many other areas of life.

\$-calculus is is a formalization of
resource-bounded computation (also called anytime algorithms),
proposed by Dean, Horvitz, Zilberstein and Russell in the late 1980s and
early 1990s \cite{horvitz01,russell95}.
Anytime algorithms are guaranteed to produce better
results if more resources (e.g., time, memory) become available.

On the other hand, the standard representative of process algebras, the $\pi$-calculus
\cite{milner89,milner99}
is believed to be the most mature approach for concurrent systems and multi-agent systems.

A unique feature of \$-calculus is support for problem solving by incrementally {\em searching} for solutions and using {\em cost} performance measure to direct its search.

Historically, \$-calculus has been ispired both by Church's $\lambda$-calulus \cite{church36} and Milner's $\pi$-calculus \cite{milner89,milner99}.
The \$-calculus rests upon the primitive notion of {\em cost} in a similar
way as the $\pi$-calculus was built around a central concept of {\em interaction} and $\lambda$-calculus around a {\em function}.
Cost and interaction concepts are interrelated in the sense that cost captures
the quality of an agent interaction with its environment.
The basic \$-calculus search method used for problem solving is called $k \Omega$-optimization.

\begin{quote}
{\em The $k \Omega$-optimization represents this ``impossible'' to construct, but ``possible to approximate indefinitely''
universal algorithm.

The advantage of the \$-calculus approach is that it allows to experiment, combine and compare various methods in a uniform framework, to experiment with them; to select on-line the best method, or perhaps, to produce a new search or machine learning algorithm. }
\end{quote}

The $k \Omega$-optimization  is a very general search method, allowing the simulation of many other search algorithms,
including A*, minimax, simulated annealing, dynamic programming, tabu search, machine learning like ID3, neural networks or evolutionary algorithms (for more details on search algorithms, see, e.g., \cite{russell95}).
Note that machine learning is based on search too, i.e., through space of functions, hypotheses, rewards, and so on. We can write briefly:  {\em Learning = Search + Memorization}. Memorization, called also memoization, distinguish learning from an amnesiac adaptive optimization that is based on search only.

Potential and implemented \$-calculus applications include machine learning and data mining, mobile robotics, cost-driven supercomputers, the cost languages paradigm, Watson-like decision systems, bioinformatics, automatic programming, automatic problem solving, big data \& cloud computing, fault-tolerant computing \& cyber security, neural networks \& deep NNs, evolutionary computation, compiler optimization.

\subsection{The \$-calculus syntax}

In \$-calculus everything is a cost expression: agents, environment,
communication, interaction links, inference engines, modified structures,
data, code, and meta-code.
\$-expressions can be simple or composite.
Simple \$-expressions $\alpha$ are considered to be executed in one atomic
indivisible step. Composite  \$-expressions $P$ consist of distinguished
components (simple or composite ones) and can be interrupted.

\begin{definition}
{\bf The \$-calculus} The set ${\cal P}$ of \$-calculus process expressions consists of simple
\$-expressions $\alpha$ and composite \$-expressions $P$, and is defined by the following syntax:

\begin{center}
\scalebox{0.98}{
\begin{tabular}{llll}
$\alpha$ & ::= & $(\$_{i \in I}\;P_i)$ & cost \\
& $|$ & $(\rightarrow_{i \in I}\;c\;P_i)$ & send $P_i$ with evaluation through channel $c$ \\
& $|$ & $(\leftarrow_{i \in I}\;c\;X_i)$ & receive $X_i$ from channel $c$ \\
& $|$ & $('_{i \in I}\;P_i)$ & suppress evaluation of $P_i$\\
& $|$ & $(a_{i \in I}\;P_i)$ & defined call of simple \$-expression $a$ with parameters $P_i$, and \\
& &  & and its optional associated definition  $(:=\; (a_{i \in I} \;X_i)\;<R>)$ \\
 & & & with body $R$ evaluated atomically\\
& $|$ & $(\bar{a}_{i \in I}\;P_i)$ & negation of defined call of simple \$-expression $a$
\end{tabular} }
\end{center}

\begin{center}
\scalebox{0.98}{
\begin{tabular}{llll}
$P$ & ::= & $( \konk_{i \in I} \;\alpha\:P_i)$ & sequential composition \\
& $|$ & $(\pa_{i \in I}\; P_i)$ & parallel composition \\
& $|$ & $(\msum_{i \in I} \;P_i)$ & cost choice \\
& $|$ & $(\psum_{i \in I} \;P_i)$ & adversary choice \\
& $|$ & $(\sqcup_{i \in I}\; \;P_i)$ & general choice \\
& $|$ & $(f_{i \in I}\;P_i)$ & defined process call $f$ with parameters $P_i$,
and its associated \\
& &  & definition $(:=\; (f_{i \in I} \;X_i)\;R)$ with body $R$ (normally
\\
& & & suppressed); $({}^1\;R)$ will force evaluation of $R$ exactly once
\end{tabular} }
\end{center}
\end{definition}

The indexing set $I$ is a possibly countably infinite.
In the case when $I$ is empty, we
write empty parallel composition, general, cost and adversary choices as  $\bot$
(blocking), and empty sequential composition ($I$ empty and $\alpha=\varepsilon$)
as  $\varepsilon$  (invisible transparent
action, which is used to mask, make invisible parts of \$-expressions).
Adaptation (evolution/upgrade)
is an essential part of \$-calculus, and all \$-calculus operators are infinite
(an indexing set $I$ is unbounded). The \$-calculus agents interact through
send-receive pair as the essential primitives of the model.

Sequential composition is used when \$-expressions are evaluated in a
textual order. Parallel composition is used when expressions run in parallel
and it picks a subset of non-blocked elements at random. Cost choice is used
to select the cheapest alternative according to a cost metric. Adversary
choice is used to select the most expensive alternative according to a cost metric.
General
choice picks one non-blocked element at random. General choice is different
from cost and adversary choices. It uses guards satisfiability. Cost and adversary choices are based on
cost functions. Call and definition encapsulate expressions in a more
complex form (like procedure or function definitions in programming
languages). In particular, they specify recursive or iterative repetition of
\$-expressions.

Simple cost expressions execute in one atomic step.
Cost functions are used for optimization and adaptation.
The user is free to define his/her own cost metrics. Send and receive perform
handshaking message-passing communication, and inferencing.
The suppression operator suppresses evaluation of the underlying \$-expressions.
Additionally, a user is free to define
her/his own simple \$-expressions, which may or may not be negated.

Note that because of infinity in indexes operators may require infinite resources (i-decidable), some can be computed
asymptotically in finite time despite infinity of indexes (a-decidable), and some use finite resources
(polynomial and exponential classes 1 and 2). We need such operators to model hypercomputational capabilities.
If indexing set is finite we have classical TM, if infinite - we can simulate  infinite cells from cellular automata,
or infinite tape of TM (what TM cannot do). Note that there is nothing decremental in that, for example,
Robin Milner's $\pi$-calculus also uses
infinity (in replication operator) and nobody objects that $\pi$-calculus is not constructable
(because in most practical cases replication is finite). Similarly, \$-calculus for effective solutions of
intractable problems uses finite indexing sets, leading to finite trees by applying the $k\Omega$-optimization. However, much
more challenging is to prove asymptotic completeness and optimality of $k\Omega$-search operating on infinite trees (a-decidability).

\subsection{The \$-calculus semantics: The $k\Omega$-search}

The operational semantics of the {\$}-calculus is defined
using the $k\Omega$-search that captures the dynamic nature and incomplete
knowledge associated with the
construction of the problem solving tree.
The basic \$-calculus problem solving method,
the $k \Omega$-optimization, is a very general search method providing
meta-control, and
allowing to simulate many other search algorithms,
including A*, minimax, dynamic programming, tabu search, or evolutionary algorithms \cite{russell95}.
The problem solving works iteratively: through select, examine and execute phases.

In the select phase  the tree of possible solutions is generated
up to $k$ steps ahead, and agent identifies its alphabet of interest for optimization $\Omega$. This
means that the tree of solutions may be incomplete in width and depth (to deal with complexity).
However, incomplete (missing) parts of the tree
are modeled by silent \$-expressions $\varepsilon$, and their cost estimated (i.e., not all information is lost).
The above
means that $k \Omega $-optimization may be if some conditions are satisfied to be complete and optimal.
In the examine phase the trees of possible solutions are pruned minimizing cost of solutions,
and in the execute phase up to $n$ instructions are executed.
Moreover, because the \$ operator may capture not only the cost of solutions, but the cost of resources used
to find a solution, we obtain a powerful tool to avoid methods that are too costly,
i.e., the \$-calculus directly
minimizes search cost. This basic feature, inherited from anytime algorithms, is needed to tackle directly
hard optimization problems, and allows to solve total optimization problems (the best quality solutions with
minimal search costs).
The variable $k$ refers to the limited
horizon for optimization, necessary due to the unpredictable dynamic nature
of the environment. The variable $\Omega $ refers to a reduced alphabet of
information. No agent ever has reliable information about all factors that
influence all agents behavior. To compensate for this, we mask factors where
information is not available from consideration; reducing the alphabet of
variables used by the \$-function. By using the $k\Omega $-optimization to find
the strategy with the lowest \$-function, meta-system finds a satisficing
solution, and sometimes the optimal one. This avoids wasting time trying to
optimize behavior beyond the foreseeable future. It also limits
consideration to those issues where relevant information is available.
Thus the $k\Omega$ optimization provides a flexible approach to local and/or
global optimization in time or space. Technically this is done by replacing
parts of \$-expressions with invisible \$-expressions $\varepsilon$, which
remove part of the world from consideration (however, they are not ignored entirely - the
cost of invisible actions is estimated).

The $k\Omega$-optimization meta-search procedure can be used both for single and multiple
cooperative or competitive agents working online ($n \neq 0$) or offline ($n=0$).
The \$-calculus programs consist of multiple \$-expressions for
several agents.

\medskip
Let's define several auxiliary notions used in the $k\Omega$-optimization meta-search.
Let:
\begin{itemize}
\item  {${\cal A}$ -} be an alphabet of \$-expression names
for an enumerable {\em universe of agent population}
(including an environment, i.e., one agent may represent an environment). Let
${\cal A} = \bigcup_i A_i$, where $A_i$ is the alphabet of \$-expression names (simple or
complex) used by the $i$-th agent, $i=1,2,...,\infty$. We will assume that the names of \$-expressions are unique, i.e.,
$A_i \cap A_j = \emptyset, i \neq j$ (this always can be satisfied by indexing \$-expression
name  by a unique agent index. This is needed for an agent to execute only own
actions).
The agent population size will be denoted by $p=1,2,...,\infty$.

\item  {$x_i[0] \in {\cal P}$ -}  be an initial \$-expression for the $i$-th agent,
and its initial search procedure $k\Omega_i[0]$.

\item
{$min(\$_i\;(k\Omega_i[t]\;\;x_i[t]))$ -} be an implicit default goal and
$Q_i \subseteq {\cal P}$ be an optional (explicit) goal.
The default goal is to find a pair of  \$-expressions,
i.e., any pair $(k\Omega_i[t],x_i[t])$
being
\[min\{(\$_i\;(k\Omega_i[t],x_i[t]))=\$_{1i}(\$_{2i}(k\Omega_i[t]),\$_{3i}(x_i[t]))\},\]
where $\$_{3i}$ is a problem-specific cost function, $\$_{2i}$ is a search algorithm
cost function,
and $\$_{1i}$ is an aggregating function combining $\$_{2i}$ and $\$_{3i}$.
This is the default goal for total optimization looking for the best solutions $x_i[t]$ with minimal search
costs $k\Omega_i[t]$.
It is also possible to look for the optimal solution only, i.e.,
the best $x_i[t]$ with minimal value of $\$_{i3}$, or the best search algorithm $k\Omega_i[t]$ with minimal costs
of $\$_{i2}$.
The default goal can be overwritten or supplemented by any other termination condition
(in the form of an arbitrary \$-expression
$Q$) like the maximum number
of iterations, the lack of progress, etc.

\item  {$\$_i$ -} a cost function performance measure
(selected from the library or user defined). It consists of the problem specific
cost function $\$_{3i}$, a search algorithm cost function $\$_{2i}$, and an aggregating
function $\$_{1i}$. Typically, a user provides cost of simple
\$-expressions or an agent can learn such costs (e.g., by reinforcement learning).
The user selects or defines also how the costs of composite \$-expressions will
be computed.
The cost of the solution tree is the function of its components:
costs of nodes (states) and edges (actions).
This allows to express both the quality of solutions
and search cost.

\item  {$\Omega_i \subseteq {\cal A}$ -} a {\em scope of
deliberation/interests of the i-th agent}, i.e., a subset of the universe's
of \$-expressions chosen for optimization. All elements of ${\cal A} -
\Omega_i$ represent irrelevant or unreachable parts of an environment, of
a given agent or other agents, and will
become invisible (replaced by $\varepsilon$), thus either ignored or
unreachable for a given agent (makes optimization local spatially).
Expressions over $A_i - \Omega_i$ will be treated as observationally
congruent (cost of $\varepsilon$ will be neutral in optimization, e.g., typically
set to 0). All expressions over $\Omega_i - A_i$ will be treated as
strongly congruent - they will be replaced by $\varepsilon$ and although
invisible, their cost will be estimated using the best available knowledge of an
agent (may take arbitrary values from the cost function domain).

\item  {$b_i = 0,1,2,...,\infty$ -} a branching factor of the search tree (LTS), i.e.,
the maximum number of
generated children for a parent node. For example, hill climbing has $b_i=1$,
for binary tree $b_i=2$, and $b_i=\infty$ is a shorthand to mean to
generate all children (possibly infinite many).

\item  {$k_i = 0,1,2,...,\infty$ -} represents {\em the depth of deliberation},
i.e., the number of steps in the derivation tree selected for optimization
in the examine phase (decreasing $k_i$ prevents combinatorial explosion, but
can make optimization local in time). $k_i=\infty$ is a shorthand to mean to
the end to reach a goal (may not require infinite number of steps). $k_i=0$
means omitting optimization (i.e., the empty deliberation) leading to reactive
behaviors. Similarly, a branching factor $b_i=0$ will lead to an empty deliberation too.
Steps consist of multisets of simple \$-expressions, i.e., a
parallel execution of one or more simple \$-expressions constitutes one
elementary step.

\item  {$n_i = 0,1,2,...,\infty$ - } the number of steps selected for
execution in the execute phase. For $n_i>k_i$ steps larger than $k_i$ will
be executed without optimization in reactive manner. For $n_i=0$ execution
will be postponed until the goal will be reached.

For the depth of deliberation $k_i=0$, the $k \Omega$-search will work in the style
of imperative programs (reactive agents), executing up to $n_i$ consecutive steps
in each loop iteration. For $n_i=0$ search will be offline, otherwise
for $n_i \neq 0$ - online.

\item {$gp$, $reinf$, $strongcon$, $update$ - } auxiliary flags used in the
$k\Omega$-optimization meta-search procedure.
\end{itemize}

Each agent has its own $k\Omega$-search procedure $k\Omega_i[t]$ used to build
the solution
$x_i[t]$ that takes into account other agent actions (by selecting its alphabet of
interests $\Omega_i$ that takes actions of other agents into account).
Thus each agent will construct its own view of the whole universe that only sometimes
will be the same for all agents (this is an analogy to the subjective
view of the ``objective'' world by individuals having possibly different goals and
different perception of the universe).

\begin{definition}
{\em
{\bf The $k\Omega$-Optimization Meta-Search Procedure}
{\em The {\em $k\Omega$- optimization meta-search procedure} $k\Omega_i[t]$
for the i-th agent}, $i=0,1,2,...$,
from an enumerable {\em universe of agent population} and
working in time generations $t=0,1,2,...$ is a complex \$-expression (meta-procedure)
consisting of simple \$-expressions $init_i[t]$, $sel_i[t]$, $exam_i[t]$,
$goal_i[t]$, $\$_i[t]$, complex \$-expression $loop_i[t]$ and $exec_i[t]$,
and constructing solutions, its input $x_i[t]$, from predefined and user defined
simple and complex \$-expressions.
For simplicity, we will skip time and agent indices in most cases if it does
not cause confusion,
and we will write
$init$, $loop$,
$sel$, $exam$, $goal_i$ and $\$_i$.
Each {\em i-th} agent performs the following
$k\Omega$-search procedure $k\Omega_i[t]$ in the time generations $t=0,1,2,...$:

\medskip
\scalebox{0.9}{
\begin{tabular}{ll}
$(:= (k\Omega_i[t]\;\;x_i[t])\; ( \konk \;
(init\;\;(k\Omega_i[0] \;x_i[0]))$ & {\small // initialize $k\Omega_i[0]$ and $x_i[0]$} \\
\hspace{.55cm}$(loop\;\;x_i[t+1]))$ & {\small // basic cycle: select, examine, } \\
$)$ &                   {\small // execute}
\end{tabular} }

\medskip\noindent
where $loop$ meta-\$-expression takes the form of the select-examine-execute cycle performing the
$k\Omega$-optimization until the goal is satisfied. At that point, the agent
re-initializes and works on a new goal in the style of the never ending
reactive program:

\medskip\noindent
\scalebox{0.9}{
\begin{tabular}{ll}
$(:=\;(loop\;\;x_i[t])$ & {\small // loop recursive definition } \\
\hspace{.55cm}$(\sqcup\;( \konk
\;\;(\overline{goal_i[t]} \;\;(k\Omega_i[t]\;x_i[t]))$ & {\small // goal not satisfied, default
goal }\\
& {\small // $min(\$_i\;(k\Omega_i[t]\;x_i[t]))$ }\\
\hspace{1.7cm}$(sel\;\;x_i[t])$ & {\small // select: build problem solution tree k step }\\
& {\small // deep, b wide} \\
\hspace{1.7cm}$(exam\;\;x_i[t])$ & {\small // examine: prune problem solution tree in }\\
& {\small // cost $\msum$ and in adversary $\psum$ choices} \\
\hspace{1.7cm}$(exec\;\;(k\Omega_i[t]\;x_i[t]))$ & {\small // execute: run optimal $x_i$ n steps
and }\\
& {\small // update $k\Omega_i$ parameters} \\
\hspace{1.65cm}$(loop\;\;x_i[t+1]))$ & {\small // return back to loop}\\
\hspace{1cm}$( \konk
\;(goal_i[t]\;\;(k\Omega_i[t]\;x_i[t]))$ & {\small // goal satisfied - re-initialize search} \\
\hspace{1.7cm}$(k\Omega_i[t]\;x_i[t])))$&\hspace{.4cm} \\
$)$ &
\end{tabular} }

Simple \$-expressions $init$, $sel$, $exam$, $goal$ with their atomicly executed
bodies are defined below. On the other hand, $exec$ can be interrupted after each
action, thus it is not atomic.

\begin{enumerate}
\item {{\bf Initialization} $(:= \;(init\;\;(k\Omega_i[0]\;x_i[0]))\;<init\_body>)$:}
where $init\_body =(\konk\;(\leftarrow_{i \in I}\;user\_channel\;X_i)$
$(\sqcup\;\overline{cond\_init} \;(\konk\;cond\_init\;
(init\_body)) )$,
and $cond\_init =(\sqcup \;(x_i[0]=\bot) \; (k_i=n_i=0))$, and  successive $X_i$,
$i=1,2,...$ will be the following:
$k\Omega_i[0]$ an initial meta-search procedure (default: as provided in this definition),
$k_i,b_i,n_i,\Omega_i,A_i$ (defaults: $k_i=b_i=n_i=\infty$, $\Omega_i= A_i={\cal A}$);
simple and complex \$-expressions definitions over  $A_i \cup \Omega_i$
(default: no definitions);\\
initialize costs of simple \$-expressions
randomly and set reinforcement learning flag $reinf=1$
(default: get costs of simple \$-expressions from the user, $reinf=0$);
$\$_{i1}$ an aggregating cost function (default: addition), $\$_{i2}$ and $\$_{i3}$
search and solution specific cost functions
(default: a standard cost function as defined in the next section);\\
$Q_i$ optional goal of computation (default: $min(\$_i\;(k\Omega_i[t],x_i[t]))$);\\
$x_i[0]$ an initial \$-expression solution (an initial state of LTS for the i-th agent)
over alphabet $A_i \cup \Omega_i$. This resets $gp_i$ flag to 0 (default: generate
$x_i[0]$ randomly in the style of genetic programming and $gp_i=1$);\\
/* receive from the user several values for initialization overwriting
possibly the defaults. If atomic initialization fails re-initialize $init$. */

\item {{\bf Goal} $(:= (goal_i[t] \, (k\Omega_i[t]\, x_i[t]))\, <goal\_body>)$:}
where $goal\_body$ checks for the~max\-imum predefined quantum of time
(to avoid undecidability or too long verification)
whether goal state defined in the $init$ phase has been reached.
If the quantum of time expires, it returns false $\bot$.

\item  {{\bf Select Phase}\\ $(:=(sel\;x_i[t])\,
<(\sqcup\, cond\_sel\_exam \, (\konk \, \overline{cond\_sel\_exam}\, sel\_body))> )$:}
where \\
$cond\_sel\_exam = (\sqcup\;(k_i=0)\;(b_i=0))$ and $sel\_body$
builds the search tree with the branching factor $b_i$ and depth $k_i$
over alphabet $A_i \cup \Omega_i$ starting from the current state $x_i[t]$.
For each state $s$ derive actions $a$ being mulitsets of simple \$-expressions,
and arriving in a new state $s'$.
Actions and new states are found in two ways:
\begin{enumerate}
\item
if $gp_i$ flag is set to 1 - by applying crossover/mutation
(in the form of send and receive operating on LTS trees) to obtain
a new state $s'$. A corresponding action between $s$ and $s'$ will be labeled
as observationally congruent $\varepsilon$
with neutral cost 0.
\item
if $gp_i$ flag is set to 0 - by applying inference rules of LTS to a state.

Each simple \$-expression in actions is labeled

\begin{itemize}
\item by its name if simple \$-expression  belongs to $A_i \cup \Omega_i$
and width and depth $b_i$,$k_i$ are
not exceeded,
\item is renamed by strongly congruent $\varepsilon$ with estimated cost
if flag $strongcong=1$ (default: renamed by weakly congruent $\varepsilon$ with
a neutral (zero) cost, $strongcong=0$) if
width $b_i$ or depth $k_i$
are exceeded /* hiding actions
outside of the agent's width or depth search horizon,
however not ignoring, but estimating their costs */.
\end{itemize}
\end{enumerate}

\noindent
For each new state $s'$ check whether width/depth of the tree is exceeded,
and whether it is
a goal state. If so, $s'$ becomes the leaf of the tree (for the current loop cycle),
and no new actions are generated,
otherwise continue to build the tree.
If $s'$ is a goal state, label it as a goal state.

\item  {{\bf Examine Phase}\\ $(:= (exam\;x_i[t])\,
<(\sqcup\, cond\_sel\_exam \, (\konk \, \overline{cond\_sel\_exam}\, exam\_body))> )$:}
where \\
$exam\_body$ prunes the search tree by selecting
paths with minimal cost in cost choices
and with maximal cost in adversary choices. Ties are broken randomly.
In optimization, simple \$-expressions  belonging to  $A_i - \Omega_i$  treat as
observationally congruent $\varepsilon$ with neutral cost
(typically, equal to 0 like e.g., for a standard cost function)
/* hiding agent's actions outside of its interests by ignoring their cost */.

\item  {{\bf Execute Phase}\\ $(:= (exec \, (k\Omega_i[t]\;x_i[t]))\,
exec\_body$):} where  \; $exec\_body = $\\
$(\konk\; (\sqcup\;
(\konk\; (n_i=0)(\overline{goal\_reached})(current\_node=leaf\_node\_with\_min\_costs))$
\newline
$(\konk\; (n_i=0)(goal\_reached)(execute(x_i[t]))(current\_node=leaf\_node))$\\
$(\konk\; (\overline{n_i=0})(execute\_n_i\_steps(x_i[t]))$\\
$(current\_node=node\_after\_n_i\_actions)))$
$update\_loop)$
\newline
/*  {\em If $n_i=0$ (offline search) and no goal state has been reached in the Select/Examine
phase there will be no execution in this cycle. Pick up the
most promising leaf node of the tree (with minimal cost) for expansion, i.e., make it a
current node (root of the subtree expanded in the next cycle of the loop
appended to an existing tree from the select phase, i.e., pruning will be invalidated
to accommodate eventual corrections after cost updates).
If $n_i=0$ (offline search) and a goal state has been reached in the Select/Examine
phase, execute optimal $x_i[t]$ up to the leaf node using a tree constructed and pruned
in the Select/Examine phase, or use LTS inference rules otherwise (for $gp=1$).
Make the leaf node a current
node for a possible expansion (if it is not a goal state - it will be a root of a new tree)
in the next cycle of the loop.
If $n_i \neq 0$ (online search), execute optimal $x_i[t]$ up to at most $n_i$ steps.
Make the last state a current state -
the root of the tree expanded in the
next cycle of the loop. In execution
simple \$-expressions  belonging to  $\Omega_i - A_i$ will be executed by other agents.}*/
\newline
The $update\_loop$ by default
does nothing (executes silently $\varepsilon$ with no cost) if $update$ flag is reset.
Otherwise if $update=1$, then
it gets from the user potentially all possible updates, e.g.,
new values of $b_i$, $k_i$, $n_i$ and
other parameters of
$k\Omega[t]$, including costs of simple \$-expressions, $\Omega_i$, $goal_i$.
If $update=1$ and the user does not provide own modifications (including possible
overwriting the $k\Omega[t]$), then self-modification
will be performed in the following way.
If execution was interrupted  (by receiving message from the user or environment
 invalidating solution found in the Select/Examine phase), then
$n_i=10$ if $n_i=\infty$, or $n_i=n_i-1$ if $n_i \neq 0$, or
$k_i=10$ if $n_i=0,k_i=\infty$, or $k_i=k_i-1$ if $n_i=0,k_i\neq\infty$.
If execution was not interrupted
increase $n_i=n_i+1$ pending $0<n_i\leq k_i$. If $n_i=k_i$ increase
$k_i=k_i+1$,$b_i=b_i+1$.
If cost of search ($\$_{2i}(k\Omega[t])$) larger than a predefined threshold
decrease $k_i$ and/or $b_i$, otherwise increase it. If reinforcement learning
was set up $reinf=1$ in the $init$ phase, then cost of simple \$-expressions will
be modified by reinforcement learning.
\end{enumerate}
}
\end{definition}

The building of the LTS tree in the select phase for $gp_i=0$ combines imperative
and rule-based/logic  styles of programming (we treat clause/production as a user-defined
\$-expression definition and call it by its name.  This is similar to Robert Kowalski's
 dynamic interpretation
of the left side of a clause as the name of procedure and the right side
as the body of procedure.).

In the $init$ and $exec/update$ phase, in fact,
a new search algorithm can be created (i.e., the old $k\Omega$ can be overwritten),
and being completely different from the original $k\Omega$-search.
The original $k\Omega$-search in self-modication
changes the values of its control parameters mostly, i.e., $k,n,b$,
but it could modify also
$goal$, $sel$, $exam$, $exec$ and \$.

Note that all parameters $k_i$, $n_i$, $\Omega_i$, $\$_i$, $A_i$, and
${\cal A}$ can evolve in successive loop iterations.
They are defined as the part of the $init$ phase, and modified/updated at
the end of the Execute phase $exec$. Note that they are associated
with a specific choice
of the meta-system: a $k\Omega$-optimization search.

\subsection{On cost performance measures and standard cost function}

The \$-calculus is built around the central notion of {\em cost}. The cost functions represent a uniform criterion of search and the quality of solutions in problem solving. Cost functions have their roots in von Neumann/Morgenstern utility theory; they satisfy axioms for utilities \cite{russell95}.
In decision theory they allow to choose states with optimal utilities on average (the maximum expected utility principle).
In \$-calculus they allow to choose states with minimal costs subject to uncertainty (expressed by probabilities, fuzzy set or rough set membership functions).

It is not clear whether it is possible to define a finite minimal and complete set of cost functions, i.e., possible to use ``for anything''.
\$-calculus approximates this desire for universality by defining a standard cost function to be open and possible to use for many things (but not for everything, thus a user may define own cost functions).
The user defines costs of atomic \$-expressions or they are randomly initialized and modified by reinforcement learning.

\begin{definition}[On standard cost function]
A {\em standard cost function} defines how to compute costs of complex \$-expressions assuming that cost of simple atomic \$-expressions are given/known, e.g.,
\begin{itemize}
\itemsep=0.9pt
\item
costs of sequential composition is defined as sum of costs of its components,
\item
cost of cost choice is equal to the cost of component with minimal cost,
\item
cost of adversary choice is equal to the cost of component with maximal cost,
\item
cost of general choice is equal to the average costs of components (if probabilities are used), or the cost of component with the maximal value of the product fuzzy set/rough set membership function  and component cost.
\end{itemize}
\end{definition}

A complete definition of standard cost function has been provided in \cite{eber07,eber08}.

Note that both $k\Omega[t]$ meta-search procedure and solution $x[t]$ take the form of \$-expressions whose costs are computed using the same standard \$-function $\$(k\Omega[t],x[t])=\$_1(\$_2(k\Omega[t]),\$_3(x[t]))$, where $\$_1$ is an aggregating function (typically addition or weighted addition), $\$_2$ and $\$_3$ are standard cost function computing costs of search and the quality of solution, respectively.

\subsection{On automatic problem solving as a multi-objective minimization search}

\begin{definition}[On search methods types]
Search methods (and $k\Omega$-optimization, in particular) can be
\begin{itemize}
\itemsep=0.9pt
\item
{\em complete} if no solutions are omitted,
\item
{\em optimal} if the highest quality solution is found,
\item
{\em search optimal} if the solution is provided with minimal search cost,
\item
{\em totally optimal} if the highest quality solution with minimal search cost is found.
\end{itemize}
\end{definition}

\begin{definition}[On cooperative and competitive search]
Search can involve single or multiple agents; for multiple agents it can be
\begin{itemize}
\itemsep=0.9pt
\item
{\em cooperative} (\$-calculus cost choice is used),
\item
{\em competitive} (\$-calculus adversary choice is used),
\item
{\em random} (\$-calculus general choice is used).
\end{itemize}
\end{definition}

\begin{definition}[On online and offline search]
Search can be
\begin{itemize}
\itemsep=0.9pt
\item
{\em offline} ($n = 0$, the complete solution is computed first and executed after without perception),
\item
{\em online} ($ n \neq 0$, action execution and computation are interleaved).
\end{itemize}
\end{definition}

\begin{definition}[On automatic problem solving as a multi-objective total minimization]
Given an objective cost function $\$ : A \times X \rightarrow  R$, $A$ is an algorithm operating on its input $X$ and $R$ set of real numbers, {\em automatic problem solving} can be understood as
as a multi-objective (total) minimization problem to find optimal  $a^* \in A_F$ and $x^* \in X_F$, $A_F \subseteq A$ are terminal states of the algorithm $A$, and $X_F \subseteq X$ are terminal states of the solution $X$ such that
\[
   \$(a^*,x^*) = min\{\$_1(\$_2(a),\$_3(x)), a \in A, x \in X\}                     \]
     where $\$_3$ is a problem-specific cost function, $\$_2$ is a search algorithm cost function, and $\$_1$ is an aggregating function combining $\$_2$ and $\$_3$.
\end{definition}

We will skip analogous definitions for optimization and search optimization problems.

\begin{definition}[On special cases of optimization problems]

\begin{itemize}
\itemsep=0.9pt
\item If $\$_1$ becomes  an identity function we obtain the {\em Pareto optimality} keeping objectives separate,
\item {\em the optimization problem} (the best quality solutions) - $\$_2$ is not used,
\item {\em the search optimization problem} (the minimal search costs) - $\$_3$ is not used,
\item {\em the total optimization problem} (the best quality solutions with minimal search costs) - uses both $\$_1$, $\$_2$ and $\$_3$.
\end{itemize}
\end{definition}

More details on the $k\Omega$-search, including inference rules of the Labeled Transition System, observation and strong
bisimulations and congruences, sufficient conditions
to solve optimization, search optimization and total optimization problems can be found in \cite{eber07}.

So far, A*, Minimax, ID3, Dining Philosophers, single/multiple local/global sequence alignment and TSP has been encoded as various instances of the same $k\Omega$-search.
We outlined how to encode as special cases of $k\Omega$-optimization: BFS, DFS, Expectiminimax, alpha-beta search, hill climbing, simulated annealing, evolutionary algorithms, tabu search, dynamic programming, swarm intelligence, neural networks, machine learning.
We demonstrated also how to encode in \$-calculus: cellular automata, random automata, neural networks, $\lambda$-calculus,   $\pi$-calculus.

\medskip
The details of implementations and applications,
can be found in \cite{eber05a,eber07}.

\subsection{\$-Calculus research, implementation and applications}

The \$-calculus research encompasses both theoretical, implementation and applications aspects. The core should be understood as a common link spanning a theory for automatic problem solving with rich applications.\vspace*{2mm}

\begin{center}
\scalebox{0.98}{
\begin{picture}(400,300)(10,-70)
\multiput(0,185)(90,0){4}{\framebox(80,40)}
\put(-6,180){\framebox(362,50)}
\put(0,185){\makebox(80,40){\shortstack{{ universality,}\\{ expressiveness,} \\{ feasibility}}}}
\put(90,185){\makebox(80,40){\shortstack{{ cost metrics}\\{ completeness}\\{ \& profiling}}}}
\put(180,185){\makebox(80,40){\shortstack{{ meta-search}\\{ algorithms}\\{ completeness}}}}
\put(270,185){\makebox(80,40){\shortstack{{ cooperation}\\{ \& competition,}\\{ uncertainty}}}}
\put(370,180){\makebox(40,40)[l]{\shortstack{{COMMON} \\{ THEORY}}}}
\multiput(20,100)(110,0){3}{\framebox(95,40)}
\put(14,95){\framebox(325,50)}
\put(370,100){\makebox(40,40)[l]{\shortstack{{ CORE} \\{ IMPLEMEN-} \\{TATION}}}}
\put(20,100){\makebox(95,40){\shortstack{{ CO\$T}  \\{\small {  compiler+interpreter}}\\{ (impl.:
flex+bison)}}}}
\put(130,100){\makebox(95,40){\shortstack{{ \$-Ruby}\\{ (implementation}\\{ in Ruby)}}}}
\put(240,100){\makebox(95,40){\shortstack{{ \$-Python}\\{ (implementation}\\{ in Python)}}}}
\put(370,0){\makebox(40,40)[l]{\shortstack{{ APPLICA-}\\{ TIONS}}}}
\multiput(0,20)(90,0){4}{\framebox(80,40)}
\put(0,20){\makebox(80,40){\shortstack{{{\small  machine learning}}\\{ \& data mining}}}}
\put(90,20){\makebox(80,40){\shortstack{{mobile}\\{ robotics}}}}
\put(180,20){\makebox(80,40){\shortstack{{ Watson-like}\\{ decision system}}}}
\put(270,20){\makebox(80,40){\shortstack{{ bioinformatics}}}}
\multiput(0,-30)(90,0){4}{\framebox(80,40)}
\put(0,-30){\makebox(80,40){\shortstack{{ automatic}\\{ programming}}}}
\put(90,-30){\makebox(80,40){\shortstack{{ big data \&}\\{\small{ cloud computing}}}}}
\put(180,-30){\makebox(80,40){\shortstack{{ neural networks}\\{ \& deep NNs}}}}
\put(270,-30){\makebox(80,40){\shortstack{{ evolutionary}\\{ computation}}}}
\put(-6,-34){\framebox(362,99)}
\put(177,180){\vector(0,-1){35}}
\put(177,180){\vector(-1,-2){17}}
\put(177,180){\vector(1,-2){17}}
\put(177,95){\vector(0,-1){30}}
\put(177,95){\vector(-1,-2){15}}
\put(177,95){\vector(1,-2){15}}
\put(177,95){\vector(-3,-2){45}}
\put(177,95){\vector(3,-2){45}}
\end{picture} }
\end{center}
\begin{center}\vspace*{-13mm}
Fig.2: Common Theory, Core Implementation and Applications\vspace*{3mm}
\end{center}

\begin{description}
\item{\bf Common Theory:}
\$-calculus as a universal model of computation and its expressiveness,
completeness of cost functions and its library,
cost function profiling and reinforcement  learning,
universality of the $k\Omega$-meta-search algorithm and meta-level control library,
cooperation and competition of multi-agent systems,
dealing with incomplete and uncertain knowledge.
\item{\bf Core Implementation:}
the CO\$T general-purpose programming language with compiler and interpreter, or
the \$-Ruby or \$-Python DSL library languages.
\item{\bf Applications:}
machine learning and data mining,
mobile robotics,
cost-driven supercomputer architectures,
cost languages paradigm,
Watson-like decision system,
bioinformatics,
automatic programming,
big data and cloud computing,
cyber security,
neural networks and deep NNs,
evolutionary computation,
compiler optimization.
\end{description}

Note that the precursor of \$-calculus - the CSA (Calculus of Self-modifiable Algorithms) SEMAL cost language \cite{eber94,horree94}, and the NAVY specialized Common Control Language \cite{buzzell05,duarte03} and Generic Behavior Message-passing Language (GBML) \cite{eber99} for communication and control of underwater autonomous robots have been implemented. The general-purpose interpreted and compiled multi-agent CO\$T language \cite{eber04}  has been designed only (with JavaCC or flex and bison for potential implementation), and two implementations of \$-cala (based on Scala) \cite{ansari13} and \$-CalcuLisp (based on Common Lisp) \cite{smith13} DSL languages were preliminary steps in implementation of $k\Omega$-search to design and implement full-blown DSL library languages \$-Ruby and \$-Python. Currently they are not implemented yet. Both \$-Ruby and \$-Python supposed to be implemented as Ruby or Python libraries, respectively, in the style of Ruby's Gems or Python's Keras, Pandas, NumPy, TensorFlow, Scikit Learn, Eli5, PyTorch, SciPy, LightGBM, Theano libraries.
Note that for implementation of \$-calculus is sufficient to implement only CO\$T compiler and interpreter, or \$-Python or \$-Ruby, e.g., in the style of Python or Ruby libraries This would allow to experiment more freely with the hypercomputational features of \$-calculus for the solutions of intractable and undecidable problems.

\section{On expressiveness of \$-calculus}

The expressiveness, i.e., the class of languages accepted by \$-calculus subsumes the expressiveness of
Turing Machines, i.e., it is straightforward (using user defined function and its call) to show how to encode in a straightforward way  $\lambda$-calculus
\cite{church36} in \$-calculus \cite{eber07}.

In fact, we will prove that \$-calculus is equally expressive as Evolutionary Automata and Interaction Machines.

To deal with undecidability,
the \$-calculus uses all three principles (not only infinity) from introductory section:
the  infinity, interaction, and evolution
principles:
the infinity -
because of the infinity of the indexing set $I$ in the \$-calculus operators,
the interaction - because
if to assume that simple \$-expressions
may represent oracles, then  we define an equivalent
to the oracle a user defined simple \$-expression,
that somehow in the manner of the ``black-box'' solves
unsolvable problems (however, we do not know how), and
the evolution - because
the $k\Omega$-optimization may be evolved to a new (and hopefully) more powerful
problem solving method (in update phase, $k\Omega$-optimization can update dynamically its components)

It is easy to think about approximate implementation of unbounded
(infinitary) concepts, than  about implementation of oracles.
This is so, because it is happening at the level of modeling and not physical implementation in a similar way as Turing allowed an infinite tape in his TM. Of course, at any moment of time
any computer has a finite memory likewise Internet consists of a finite number of nodes.
 It is not clear how to implement oracles (as Turing stated
{\em an oracle cannot be a machine}, i.e., implementable by mechanical means classical recursive algorithm),
and as the result,
the models based on them. One of potential implementation of oracles could be an
infinite lookup table with all results of the decision for the machine \cite{kozen97} stored. The quite
different story is how to initialize such infinite lookup table and how to search
it effectively
using for instance hashing or specialized data structures (again model versus physical implementation).

\begin{theorem}[On expressiveness of \$-calculus - version I]
The terminal language for \$-calculus agents coincides with the class of all languages in the alphabet~$X$.
\end{theorem}

\begin{proof}
We generally repeat the proof for Evolutionary Automata or Interaction Machines.
We construct an infinite parallel composition - \$-calculus
\$-expression $(\pa_i \; A_i)$, where $A_i$ represents the $i$-th agent, and put to it an arbitrary input string $w$.
If $(\pa_i \; A_i)$  accepts $w$ (i.e., one of $A_i$ accepts) then \$-calculus accepts. If no $A_i$
accepts then $L_d$ does not accept either.
We show that given a formal language $L$, i.e., a set of finite words in the alphabet $X$, there is a \$-calculus set of agents $A$ such that $A$ accepts $L$. We assume that all components/agents
of $A = \{A[t]; t = 0, 1, 2, 3,\ldots\}$ work in sequential order corresponding to their indices in generations $t = 0, 1, 2, 3,\ldots$. Each agent of $A$ takes as input a string and
produces as output a string to pass it to the successor (if it does not accept) or stops the computation if the string is accepted.
To do this, for each word $w$, we build an agent (e.g., in the form of the finite automaton) $A_w$ that given a word $w$ as its input, it accepts only the word $w$, and given any other word $u$, it outputs $u$,
which goes as input to the next agent in the \$-calculus $A$.
In such a way an agent interacts with its predecessor and its successor. In both cases, the agent $A_w$ comes to a terminal state. Then taking any sequence
$E = \{A[t] = A_w ; w \in L\}$  of such set of agents, we obtain the necessary \$-calculus $A$.
\end{proof}

\begin{theorem}[On expressiveness of \$-calculus - version II]
The terminal language for \$-calculus agents coincides with the class of all languages in the alphabet~$X$.
\end{theorem}

\begin{proof} We could assume that oracles can accept aribitrary languages over alphabet $X$ (the same assumption as was done for Turing o-machines and Site and Internet Machines. We could assume further that \$-calculus atomic \$-expressions may represent black box oracles, i.e., can accept any set of strings in the alphabet $X$. However, we will do something simpler instead. Let's asssume very primitive oracles of the class of finite automata, and let an infinite set of oracles correspond to \$-calculus agents. After that the proof is like in previous theorem.
\end{proof}

\section{On completeness of cost  metrics}

A unique aspect of the \$-calculus, compared to other models of computation, is its associated cost mechanism. A standard cost function is predefined in \$-calculus using either probabilities, fuzzy sets or rough sets membership functions to express uncertainties. The user is free to define own cost functions.

We try to address an open research question what is the minimal and complete set of cost function metrics? It will be desirable to establish a library of most common cost metrics spanning as many applications as possible. Such a standardization attempt should benefit also the field of evolutionary computation or autonomous agents.

\begin{definition}[On o-completeness of cost functions]
The set of cost functions is open-complete (o-complete) if we can expand it arbitrarily by adding or modifying cost functions.
\end{definition}

\begin{theorem}[On o-completeness of the \$-calculus standard cost function]
The \$-calculus standard cost function is o-complete.
\end{theorem}

\begin{proof}
The proof follows directly from the definition of o-completeness and the definition of the standard cost function in \$-calulus that is o-complete.
\end{proof}

Unfortunately, the openess of the standard cost function in \$-calculus causes that it is not minimal, and for sure, not finite.

Thus the problem remains whether is possible to define a FINITE minimal and complete set of cost functions, i.e., whether is possible to define better and more general axioms for costs than are provided, for instance, by the utility theory (e.g., in a style of Kolmogorov's axioms from the probability theory \cite{russell95}).

Note also that the cost functions should be expanded to deal better with infinite cases, i.e., if the cost of an infinite path is finite, i.e., forming a convergent series, we have a finite measure to guide the search, however, if the cost is infinite, we have only one value of infinity that is not useful in the search. In other words, cost functions should express different types of infinity.

\medskip
Note also that the cost function profiling is a crucial problem for applicability of the \$-calculus, i.e., how to establish the way to define costs of atomic cost expressions (an analogue of the elementary probabilities distributions). Resource-bounded optimization typically assumes that statistics from previous experiments are used to develop profiles of cost functions. Reinforcement learning can be used to derive costs of atomic and composite cost expressions. We can employ directly Q-learning or Temporal Difference learning \cite{russell95}. The advantage of reinforcement learning is that it suffices for an agent who very often does not know (or is unable to find) how to model other agents and/or an environment. Then the costs can be learned directly from reward feedback interactions.

\section{On completeness of  meta-search algorithms}

The \$-calculus uses the $k\Omega$-optimization for meta-level control, which is a generalization of control structures used for adaptive expert systems, machine learning, evolutionary computation and neural nets. Truly, this is an approximation of the non-existing universal search algorithm. However, in a similar way as an evolutionary algorithm may lead to a more complex form like the cultural, or genocop algorithm \cite{russell95}, the \$-calculus allows for an arbitrary new meta-control. Can we use an ``imperfect'' (in the sense that it can be improved, i.e., not being optimal) the $k\Omega$-optimization to obtain a better version of itself (i.e., something in the style of bootstrapping compilers allowing to obtain better versions of themselves)?
In other words, is $k\Omega$-optimization complete?

\medskip
We are getting a weaker result that the $k\Omega$-optimization can be indefinitely improved.

\begin{definition}[On o-completeness of meta-search algorithms]
The set of meta-search algorithms is open-complete (o-complete) if we can expand it arbitrarily by adding or modifying meta-search algorithms.
\end{definition}

\begin{theorem}[On o-completeness of the $k\Omega$-optimization meta-search algorithm]
The \$-calculus $k\Omega$-optimization meta-search algorithm is o-complete.
\end{theorem}

\begin{proof}
The proof is the following. We put to the input of
the $k\Omega$-optimization the $k\Omega$-optimization itself. If sufficient conditions for optimization are satisified, i.e., completeness of search, elitist selection and no-overestimates of missing parts of the tree, the $k\Omega$-optimization will monotically improve reaching minimal costs. This process may be infinite because the cost function may fluctutae in the process as well. The above guarantees o-completeness of the $k\Omega$-optimization.
\end{proof}

Assuming only o-completeness of the $k\Omega$-optimization, the  work on population of the library of typical meta-control algorithms is needed.

Note that in \cite{eber21}, we proposed a similar approach to derive the optimal theory for machine learning. For example, if to encode various machine learning algorithms as \$-calculus examples in the form of \$-expressions and to use the $k\Omega$-optimization, we can induce automatically by learning from examples the optimal machine learning theory from the hypothesis space of all possible theories (at least hypothetically, and, hopefully, in practice too).

The open problem remains whether we can define (similar like for cost functions) a FINITE complete and minimal set of meta-search algorithms.

\section{Conclusions and future work}

In the paper we defined three complexity classes for TM undecidable problems: U-complete, D-complete and H-complete classes. Now, it's time to populate those classes in the style like Cook, Levin and Karp did with the NP-complete class. Note also that similar like for the NP-complete class, it is enough to prove a decidability of one U-complete, or D-complete or H-complete problem instance to break the whole hierarchical puzzle
of unsolvability. At least inside of those classes. However, we do not even know whether other undeciable outside of those classes exist.  For example, solving the halting problem of UTM is pivotal to solve many other Turing Machine unsolvable problems, including to decide
$L_{ne}$, nontrivial properties,
PCP, BBP and ECP.

We investigated also expressiveness of several super-Turing models of computation, concentrating mostly on \$-calculus expressiveness, and completeness of its cost function metrics and meta-search algorithms. It looks that \$-calculus may provide a unifying theory of machine learning \cite{eber21}.

\begin{quote}
Turing Machines can be considered as an attempt to create the {\em theory of everything} for computer science, whereas similar attempts of complete theories for physics (by Newton, Laplace, Einstein, Hawking), mathematics (by Hilbert, G\"{o}del, Church, Turing) or philosophy (by Aristotle, Plato, Hegel) have failed \cite{wegner12}.  If Turing machines were truly complete, computer science with its Turing machine model would be an exception from other sciences, and computer science together with its Turing machine model would be complete. If so, by reduction techniques, we could prove also completeness of mathematics (decision problem in mathematics \cite{whitehead10} - disproved by G\"{o}del, Church and Turing \cite{godel31,church36,turing36}), and completeness of physics, philosophy, medicine, economy and so on.
\end{quote}

Obviously, the \$-calculus and other super-Turing models of computation, although being more complete than Turing machines,  do not constitute the {\em theory of everything} either. The search for such complete theory of everything will be never ending process improving indefinitely our state of knowledge (assuming its o-completeness (open-completeness)).

We are aware that we only ``scratched'' the area of super-Turing computation and super-recursive algorithms in this paper. We believe that we will witness a new area of fruitful research for many years to come. Some promising hints can be seen already in \cite{aho11,denning11}, or algorithms that run forever \cite{kleinberg06} to achieve a more powerful and more complete theory of computation.

\subsubsection*{\bf Acknowledgments} The author would like to thank anonymous reviewers for their comments and criticism. In particular, comments by Reviewer 3 have been found to be the most useful and constructive.


\begin{thebibliography}{99}
\bibitem{aho11} Aho A.
 Computation and Computational Thinking, Ubiquity ACM symposium ``What is computation'', Ubiquity Magazine, Vol.2011, Issue January, Article no.1, Jan. 2011, doi:10.1145/1922681.1922682.

\bibitem{ansari13} Ansari B. \$-cala: An Embedded Programming Language Based on \$-Calculus and Scala, Master Thesis, Rensselaer Polytechnic Institute at Hartford, Dept.of Eng. and Science, 2013.

\bibitem{bringsjord12} Bringsjord SG, Naveen S, Eberbach E, Yang Y.
 Perhaps the Rigorous Modeling of Economic Phenomena Requires Hypercomputation, Int. Journal of Unconventional Computing,
 vol.8, no.1, 2012, 3-32.

\bibitem{burgin05} Burgin M. Super-Recursive Algorithms, Springer-Verlag, 2005.  doi:10.1007/b138114.

\bibitem{burgin12} Burgin M, Eberbach E.
 Evolutionary Automata: Expressiveness and Convergence of Evolutionary Computation, Computer Journal, 55(9), 2012, 1023-1029,
 doi:dx.doi.org/10.1093/comjnl/bxr099.

\bibitem{buzzell05} Buzzell Ch.
 A Common Control Language for Multiple Autonomous Undersea Vehicle Cooperation, Master Thesis, Univ. of Massachusetts
 Dartmouth, Intercampus Graduate School of Marine Sciences and Technology, 2005.

\bibitem{church36} Church A. An Unsolvable Problem of Elementary Number Theory,
American Journal of Mathematics, vol.58, 1936, 345-363.
URL \url{http://links.jstor.org/sici?sici=0002-9327%28193604%2958%3A2%3C345%3AAUPOEN%3E2.0.CO%3B2-1}

\bibitem{denning11} Denning P.
 What Have We Said About Computation? Closing Statement, Ubiquity ACM symposium ``What is computation'', Ubiquity Magazine, Vol.2011, Issue April, Article no.1, April 2011, doi:10.1145/1967045.1967046.

\bibitem{duarte03} Duarte Ch, Martel G, Eberbach E, Buzzell Ch.
 A Common Control Language for Dynamic Tasking of Multiple Autonomous Vehicles, Proc. of the 13th Intern. Symp. on Unmanned Untethered Submersible Technology UUST'03, Durham, NH, August 24-27, 2003.

\bibitem{eber94} Eberbach E.
 SEMAL: A Cost Language Based on the Calculus of Self-modifiable Algorithms, Intern. Journal of Software Engineering and Knowledge Engineering, vol.4, no.3, 1994, 391-408.    doi:10.1142/S0218194094000192.

\bibitem{eber99} Eberbach E, Phoha S.
 SAMON: Communication, Cooperation and Learning of Mobile Autonomous Robotic Agents, Proc. of the 11th IEEE Intern. Conf. on Tools with Artificial Intelligence ICTAI'99, Chicago, IL, 1999, 229-236. doi:10.1109/TAI.1999.809791.

\bibitem{eber03c} Eberbach E, Wegner P.
 Beyond Turing Machines, The Bulletin of the European Association for Theoretical Computer Science (EATCS Bulletin),
81, Oct. 2003, 279-304.

\bibitem{eber04a} Eberbach E, Goldin D, Wegner P.
 Turing's Ideas and Models of Computation, in: (ed. Ch.Teuscher) Alan Turing: Life and Legacy of a Great Thinker, Springer-Verlag, 2004, 159-194.
 doi:10.1007/978-3-662-05642-4\_7.

\bibitem{eber04} Eberbach E, Eberbach A.
 On Designing CO\$T: A New Approach and Programming Environment for Distributed Problem Solving Based on Evolutionary Computation and Anytime Algorithms, Proc. 2004 Congress on Evolutionary Computation CEC'2004, vol.2, Portland, Oregon, 2004, 1836-1843.

\bibitem{eber05a} Eberbach E.
 \$-Calculus of Bounded Rational Agents: Flexible Optimization as Search under Bounded Resources in Interactive Systems, Fundamenta Informaticae, vol.68, no.1-2, 2005, 47-102.

\bibitem{eber05b} Eberbach E.
 Toward a Theory of Evolutionary Computation, BioSystems, vol.82, no.1, 2005, 1-19.
 doi:10.1016/j.biosystems.2005.05.006.

\bibitem{eber07} Eberbach E.
 The \$-Calculus Process Algebra for Problem Solving: A Paradigmatic Shift in Handling Hard Computational Problems, Theoretical Computer Science, vol.383, no.2-3, 2007, 200-243. doi:10.1016/j.tcs.2007.04.012.

\bibitem{eber08} Eberbach E.
 Approximate Reasoning in the Algebra of Bounded Rational Agents, Intern. Journal of Approximate Reasoning, vol.49, issue 2, 2008, 316-330
 doi:10.1016/j.ijar.2006.09.014.

\bibitem{eber15} Eberbach E.
 On Hypercomputation, Universal and Diagonalization Complete Problems, Fundamenta Informaticae, IOS Press, 139 (4), 2015, 329-346,  doi:10.3233/FI-2015-1237.

\bibitem{eber21} Eberbach E, Strzalka D.
 In Search of Machine Learning Theory, Proc. Future Technologies Conference FTC'2021, Springer, Lect. Notes in Networks and Systems,
 Oct. 28-29, 2021, Vancouver, BC, \texttt{https://saiconference.com/FTC2021}.

\bibitem{eber22} Eberbach E.
 Undecidability and Complexity for Super-Turing Models of Computation, Proc. Conf. on Theoretical and Foundational Problems in Information Studies, IS4SI 2021, The 2021 Summit of the Intern. Society for the Study of Information, Sept. 12-19, 2021, Proceedings 2022, 81, 123, doi:10.3990/proceedings2022081123.

\bibitem{godel31} G\"{o}del K.
 \"{U}ber formal unentscheidbare S\"{a}tze der Principia Mathematica und verwander Systeme, Monatschefte f\"{u}r
 Mathematik und Physik, 38:173-198, 1931.

\bibitem{hopcroft01} Hopcroft JE, Motwani R, Ullman JD.
 Introduction to Automata Theory, Languages, and Computation, 3rd ed., Addison-Wesley, 2007. ISBN-13:9780321455369.

\bibitem{horree94} Horree S.
 Design and Implementation of a Version of SEMAL Interpreter, Honors Thesis, Jodrey School of Computer Science, Acadia University, 1994.

\bibitem{horvitz01} Horvitz E.
 Zilberstein S. (eds.), Computational Tradeoffs under Bounded Resources, Artificial Intelligence 126, 2001, 1-196.

\bibitem{kleinberg06} Kleinberg J, Tardos E.
 Algorithm Design, Pearson/Addison Wesley, 2006.  ISBN-13:9780137546350.

\bibitem{kozen97} Kozen DC. Automata and Computability, Springer-Verlag, 1997.

\bibitem{kuratowski77} Kuratowki K. Introduction to Set Theory and Topology, PWN, Warsaw, 1977.

\bibitem{milner89}  Milner R, Parrow J, Walker D.
 A Calculus of Mobile Processes, Rep. ECS-LFCS-89-85 and -86, Lab. for Foundations of Computer Science,
Computer Science Dept., Edinburgh Univ., 1989.

\bibitem{milner99} Milner R.
 Communicating and Mobile Systems: The $\pi$-calculus, Cambridge University Press, 1999.  ISBN:9780521658690.

\bibitem{pesonen11} Pesonen V. New Models of Computation, 2011,
\texttt{https://slideplayer.com/slide/11412947/}.


\bibitem{post46} Post E.  A Variant of a Recursively Unsolvable Problem, Bulletin of the AMS 52, 1946, 264-268.

\bibitem{rado62} Rado T.
 On Non-Computable Functions, Bell System Technical Journal, vol.41, no.3, 1962, 877-884.
  doi:10.1002/j.1538-7305.1962.tb00480.x.

\bibitem{rice53} Rice HG.
 Classes of Recursively Enumerable Sets and their Decision Problems, Trans. of the AMS 89, 1953, 25-59.
 doi:10.1090/S0002-9947-1953-0053041-6.

\bibitem{russell95}  Russell S, Norvig P.
 Artificial Intelligence: A Modern Approach, Prentice-Hall, 3rd ed., 2010.
 ISBN-10:0136042597, 13:978-0136042594.

\bibitem{smith13} Smith J.
 \$-CalcuLisp, an Implementation of \$-Calculus Process Algebra in Common Lisp, Master Project, Rensselaer Polytechnic Institute at Hartford,
 Dept.of Eng. and Science, 2013.

 \bibitem{syropoulos08} Syropoulos A.
  Hypercomputation: Computing Beyond the Church-Turing Barrier, Springer-Verlag, 2008.   doi:10.1007/978-0-387-49970-3.

\bibitem{turing36} Turing A.
 On Computable Numbers, with an Application to the Entscheidungsproblem, Proc. London Math. Soc., 42-2, 1936, 230-265;
 A correction, ibid, 43, 1937, 544-546.    doi:10.1112/plms/s2-42.1.230.

\bibitem{turing39} Turing A.
 Systems of Logic based on Ordinals, Proc. London Math. Soc. Series 2, 45, 1939, 161-228.  doi:10.1112/plms/s2-45.1.161.

\bibitem{turing48} Turing, A.
Intelligent Machinery, 1948, in Collected Works of A.M. Turing: Mechanical Intelligence, ed.D.C.Ince, Elsevier Science, 1992.

\bibitem{wegner12} Wegner P, Eberbach E, Burgin M.
 Computational Completeness of Interaction Machines and Turing Machines, Proc. The Turing Centenary Conference, Turing-100,
 Alan Turing Centenary, EasyChair Proc. in Computing, EPiC vol. 10 (ed. A. Voronkov), Manchester, UK, June 2012, 405-414.

\bibitem{whitehead10} Whitehead AN, Russell B. Principia Mathematica, vol.1, 1910, ISBN-10:1603861823, 13:978-1603861823.
vol.2, 1912,  \texttt{https://www.slideshare.net/Kadio3/insurance-250661884},
vol.3, 1913, Cambridge Univ. Press, London.
\end{thebibliography}
\end{document}